\documentclass[12pt]{article}

\usepackage[utf8]{inputenc}
\usepackage{amsfonts}
\usepackage{amsthm}
\usepackage{amsmath}
\usepackage{ulem}

\newcommand{\Z}{\mathbb{Z}}
\newcommand{\N}{\mathbb{N}}
\newcommand{\C}{\mathbb{C}}

\newcommand{ \del }{\ | \ }
\newcommand{ \ndel }{\mbox{\ \xout{$\mid$} \ }}

\newcommand{\suml}{\sum\limits}
\newcommand{\prodl}{\prod\limits}
\newcommand{\bra}[1]{ \langle #1 |}
\newcommand{\ket}[1]{ | #1 \rangle }
\newcommand{\bkScal}[2]{ \langle #1 | #2 \rangle }

\newcommand{\abk}[1]{\left( #1 \right)}
\newcommand{\abks}[1]{\left[ #1 \right]}
\newcommand{\abkf}[1]{\left\{ #1 \right\}}
\newcommand{\abkm}[1]{\left| #1 \right|}

\makeatletter
\newtheorem*{rep@theorem}{\rep@title}
\newcommand{\newreptheorem}[2]{%
\newenvironment{rep#1}[1]{%
 \def\rep@title{#2 \ref{##1}}%
 \begin{rep@theorem}}%
 {\end{rep@theorem}}}
\makeatother

\newtheorem{theorem}{Theorem}
\newtheorem{lemma}{Lemma}

\newtheorem{definition}{Definition}
\newtheorem{propos}{Proposition}
\newtheorem*{remark}{Remark}
\newtheorem{coroll}{Corollary}

\newreptheorem{theorem}{Theorem}	
\newreptheorem{lemma}{Lemma}

\newcommand{\CCay}[1]{$ \mbox{\textit{Cay}} \abk{#1} $}
\newcommand{\Cay}[2]{$ \mbox{\textit{Cay}} \abk{#1, #2} $}
\newcommand{\ccay}[1]{\mbox{\textit{Cay}} \abk{#1}}
\newcommand{\cay}[2]{\mbox{\textit{Cay}} \abk{#1, #2}}

\newcommand{\ShiftOp}{\hat{Q}}
\newcommand{\CoinOp}{\hat{C}}
\newcommand{\WalkOp}{\hat{U}}

\usepackage{authblk}
\author{Ilnur~Khuziev\thanks{ilnur.khuziev@yandex.ru}}
\affil{Scientific Adviser: Mikhail Vyalyi}
\affil{Chair of data analysis, DIHT, MIPT}

\title{Quantum walk in symmetric Cayley graph over $\Z_2^n$}

\begin{document}
	\maketitle
	
	\begin{abstract}
		We show that the hitting time of the discrete quantum walk on a symmetric Cayley graph over $\Z_2^n $ from a vertex to its antipodal is polynomial in degree of the graph. We prove that returning time of quantum walk on a symmetric Cayley graph over $\Z_2^n $ is polynomial and the probability to hit is almost one. To prove it, we give a new estimation of Kravchuk coefficients. We give an example of a probabilistic polynomial algorithm that finds an antipodal vertex in symmetric Cayley graphs.
	\end{abstract}
	
	\section*{Introduction}
		~
		
		A discrete quantum walk is a generalisation of a random walk. In Section~\ref{sec:qwDef} we define a discrete quantum walk and give a diagonalization of a quantum walk operator for special cases.

		In random walk the probability to be in a particular vertex after sufficient large steps is very small. But in quantum walk may exist times and vertices such that the probability to be in a particular vertex is high. This property is called hitting.
		
		Julia Kempe in \cite{kempe} showed that hitting time of quantum walk on the hypercube from one corner to its opposite is polynomial and the probability to hit is $ 1 -o(1) $. In this paper we generalise her results for Cayley graph over $ \Z_2^n $ with generating set $ \abkf{e \in \Z_2^n \Big| |e|=s} $, where $ s $ is a positive integer.
				
		\begin{definition}
			\Cay{G}{S} --- \textbf{Cayley graph} over a group $G$ with a generating set $S$ is a graph $(V,E)$ with $V=G$ and $(v_1, v_2) \in E \Leftrightarrow \exists e \in S:\ e v_1 = v_2$. We will suppose that $S^{-1} = S$ (the graph is undirected) and $1 \notin S$ (vertices have no loops).
			
			We denote $ \cay{\Z_2^n}{S} $ with $ S = \abkf{e \in \Z_2^n: |e| = s} $ as \CCay{s}.
		\end{definition}
						
		In Section $ \ref{sec:retTime} $ we show that exists a moment of time $ T $ such that a quantum walk returns into initial vertex with probability almost one. We prove  necessary and sufficient conditions for hitting and moment $ T_p(s,n) $.
		
		To prove this two results we use the technique of Julia Kempe and the new bound on Kravchuk coefficients. Also we use the Lucas' theorem to analyse evenness of Kravchuk coefficients.
		
		~
		
		In paper \cite{krasin} Krasin found automorphisms group of symmetric Cayley graphs for large set of parameters. In Section \ref{sec:antipod} we generalise his results and present new properties of automorphisms groups that we use in Section \ref{sec:problem}.
		
		~
		
		It was proved in the paper \cite{expFast} that there exist graphs such that every classical algorithm traverse it exponentially slower than quantum walk. In Section \ref{sec:antipod} we introduce the problem of an antipodal vertex search. This problem generalises the problem from \cite{expFast}.
		
		Using results of Section \ref{sec:antipod} we prove that if a pair $ (s,n) $ is satisfy conditions of Theorem \ref{th:hitTime} then quantum walk solves the problem on the symmetric Cayley graph (Section \ref{sec:quantAlg}).
		
		In section \ref{sec:clAlg} we give a probabilistic algorithm that solves the problem on symmetric Cayley graphs with comparable efficiency. In other words, we show that our bounds don't give exponential speed up.
		
		~
		
		Also we generalise some results of Julia Kempe about measured quantum walk. They can be found in the Appendix \ref{sec:measuredApp}.
		
	\section{Quantum walk} \label{sec:qwDef}
		In this section we define a discrete quantum walk for a large class of graphs. Then we give a special form of a quantum walk operator in case of a Cayley graph over an Abelian group. In the last part of this sections we find spectrum of a quantum walk operator in case of a Cayley graph over $ \Z_2^n $ for a special coin operator.
		
		\subsection{Definition of quantum walk}
			\begin{definition}
				A graph is said to be \textbf{regular} if all it's vertices have the same degree.
			\end{definition}

			\begin{definition}
				Let $B = \abkf{1, 2, \ldots, m}$. Function $f: B \times V \to V$ is called an \textbf{invertible colouring} of a $ m $-regular graph $G=(V,E)$ if the following conditions hold:
				 	\begin{enumerate}
				 		\item
				 		$f(b, v_1) = f_b(v_1) = v_2 \Rightarrow (v_1, v_2) \in E $
				 		\item
				 		$\forall b \in B \ f_b$ is a permutation of the set $V$. 
				 	\end{enumerate}
			\end{definition}

			\begin{remark}
				Every Cayley graph \Cay{G}{S} has the canonical invertible colouring in $|S|$ colours, generated by the rule $f(b,v) = s_b \circ v, \ s_b \in S$. 
			\end{remark}
			
			Let $G=(V,E)$ be a $m$-regular graph, $ f $ be a invertible colouring of $ G $.
			
			Let $ X = \C^{m} \otimes \C^{|V|} $ be a $m|V|$ dimensional complex space. We will call the first tensor factor \textbf{coin-space} and associate it's basis vectors with set $B$. The second tensor factor is called \textbf{main-space}, it's basis vectors are associated with the set $V$. The basis states (vectors) of the space are denoted by $\ket{b,v}=\ket{b} \otimes \ket{v} = \ket{b}
\ket{v} $.
			
			\begin{definition}
				The \textbf{shift operator} acts on the basis states by the rule:
					\begin{equation*}
						\ShiftOp \ket{b,v} = \ket{b, f_b(v) }.
					\end{equation*}
				An action of $ \ShiftOp $ on the whole $ X $ is defined by linearity.
			\end{definition}
			
			\begin{propos}
				The shift operator is an unitary operator.
			\end{propos}			
			\begin{proof}
				By the definition of invertible colouring $\ShiftOp$ acts on the set \linebreak $\abkf{ \ket{b, v_1}, \ldots, \ket{b, v_|V|} }$ as a permutation, therefore it acts as a permutation on the set
					$$
					\bigcup\limits_{b \in B}\abkf{\ket{b, v_1 }, \ldots, \ket{b, v_|V| } }.
					$$
			\end{proof}
			
			\begin{propos}
				The shift operator has the following form:
				\begin{eqnarray*}
					\ShiftOp = \suml_{b \in B} { \ket{b} \bra{ b } \otimes \hat{Q}_b \ }, \\
					\ShiftOp_b \ket{v} = \ket{ f_b(v) }.
				\end{eqnarray*}
			\end{propos}
			
			Let's fix an arbitrary unitary operator $\CoinOp: \C^m \to \C^m$. It will be called the \textbf{coin operator}.
			
			\begin{definition}
				\textbf{The quantum walk operator} for the graph $G$ with invertible colouring $f$ is
					\begin{equation*}
					\WalkOp = \ShiftOp \circ ( \CoinOp \otimes \hat{I}),
					\end{equation*}
				where $\hat{I}$ is the identity operator on the main space.
			\end{definition}
			
			\begin{definition}[Quantum walk]
				Quantum walk state in moment $t \in \N$ with initial state $\ket{\psi_0}$ is
				\begin{equation*}
					\ket{\psi_t } = \WalkOp^t \ket{\psi_0}.
				\end{equation*}
			\end{definition}
			
			In other words, at each step we "toss" the coin and go along the direction that it pointed to.
			
			\begin{definition}[Symmetric initial state]
				If $\ \exists v \in V$ such that the initial state $\ket{\psi_0}$ is presented as follows:
				\begin{eqnarray*}
					\ket{\Psi} =  \frac{1}{\sqrt{m}}\suml_{j=1}^m {\ket{j}}, \\
					\ket{\psi_0} = \ket{\Psi} \ket{v},
				\end{eqnarray*}
				then $\ket{\psi_0}$ is called symmetric initial state with starting vertex $v$.
			\end{definition}
			
			\begin{definition}
				Projection operator on the state $ \ket{x} $:
				\begin{equation*}
					\Pi_x = \hat{I}_m \otimes \ket{x} \bra{ x }.
				\end{equation*}
			\end{definition}
			
			\begin{definition}
				The \textbf{probability to hit} vertex $\ket{x}$ at the moment $t$ with the initial state $\ket{\psi}$ is
				\begin{equation}
					\| \Pi_x \circ \hat{Q}^t \ket{ \psi_0 } \| ^ 2. \label{eq:hitprobDef}
				\end{equation}
				If an initial state has a starting vertex $x$, then the probability (\ref{eq:hitprobDef}) is called \textbf{the probability to return}.
			\end{definition}
			
		\subsection{Quantum walk in a Cayley graph over Abelian group}
			
			Let $G, |G| = N$, be a finite Abelian group
			\begin{eqnarray*}
				G = \Z_{n_1} \times \Z_{n_2} \times \cdots \times \Z_{n_{h-1}} \times \Z_{n_h}.
			\end{eqnarray*}
			
			We associate elements of $G$ with their coordinate presentation $$\ket{v} = \ket{v_1, \ldots, v_h}. $$
			
			\begin{definition}
			The \textbf{Fourier transform} of the basis vector $v$ is 
				\begin{eqnarray*}
					\ket{ \tilde{v} }= \ket{ \widetilde{ v_1,\ldots, v_h } } = 
					\suml_{k = (k_1,\ldots, k_h) \in G} {{ \frac{1}{\sqrt{N}}
					\abk{ \prod\limits_{j=1}^h{
					e^{
					-2 \pi i \frac{ k_j }{n_j} v_j
					}}
					}
					\ket{ k }}}.
				\end{eqnarray*}
			\end{definition}
			
			\begin{lemma} \label{lem:abel_pres}
				$\WalkOp$ can be presented in the following way:
				\begin{equation*}
				\WalkOp = \suml_{v \in V} { \CoinOp_v \otimes \ket{ \tilde{v} } \bra{ \tilde{v} }},
				\end{equation*}
				where
				\begin{equation*}
				\CoinOp_v = 
					\abk{ \begin{array}{cccc}
					\prodl_j e^{  2 \pi i \frac{ v_j (e_1)_j }{n_j}} & & & \\
					& \prodl_j e^{  2 \pi i \frac{ v_j (e_2)_j }{n_j}} & & \\
					& & \ddots & \\
					& & & \prodl_j e^{  2 \pi i \frac{ v_j (e_m)_j }{n_j}} 
					\end{array}
					} \circ \CoinOp .
				\end{equation*}
				Here $ e_k $ denotes the generating element associated with colour $k$.
			\end{lemma}
			
			\begin{proof}
			By direct computation we have:
			\begin{eqnarray*}
				&\ShiftOp \ket{ b } \ket{ \tilde{v} } = 
					\suml_{k = (k_1,\ldots, k_h) \in G} { \frac{1}{\sqrt{N}}
						\abk{ \prodl_{j=1}^h{
							e^{-2 \pi i \frac{ k_j }{n_j} v_j}
						}}
						\ShiftOp \ket{ b } \ket{ k }} = &\\
				&\suml_{k = (k_1,\ldots, k_h) \in G} { \frac{1}{\sqrt{N}}
					\abk{ \prodl_{j=1}^h{
						e^{-2 \pi i \frac{ k_j }{n_j} v_j}
					}}
					\ket{ b } \ket{ k+e_b }} = &\\
				&\suml_{k = (k_1,\ldots, k_h) \in G} { \frac{1}{\sqrt{N}}
					\abk{ \prodl_{j=1}^h{
						e^{-2 \pi i ( \frac{ k_j + (e_b)_j}{n_j} + \frac{ -(e_b)_j }{n_j} ) v_j}
					}}
					\ket{ b } \ket{ k+e_b }} = & \\
				&
					\abk{ \prodl_{j=1}^h{
						e^{2 \pi i \frac{ v_j (e_b)_j }{n_j} } 
					} \ket{ b }} 
					\suml_{k = (k_1,\ldots, k_h) \in G} { \frac{1}{\sqrt{N}}
						\abk{\prod\limits_{j=1}^h{
							e^{-2 \pi i \frac{k_j}{n_j} v_j}
						}}
						\ket{ k }} = &\\&
				\abk{ 
				\prodl_{j=1}^h{e^{2 \pi i \frac{ v_j (e_b)_j }{n_j}}} \ket{ b }
				}
				\otimes \ket{ \tilde{v} }.&
			\end{eqnarray*}
			\end{proof}
			
		\subsection{Quantum walk in \Cay{\Z_2^n}{S} } \label{sec:qwZ2}
		
		$ \Z_2^n \cong \mathbb{F}_2^n \cong \abkf{0,1}^n $  has the inner product:
			$$ (v',v'') = \abk{\suml_{k=1}^n v_k' \times v_k''} \pmod{2}.  $$
		
		In $G = \Z_2^n$ the Fourier transform is the Hadamard transform:
		\begin{eqnarray}
			\ket{ \tilde{v} }= \ket{ \widetilde{ v_1,\ldots, v_n } } = 
			{ \frac{1}{\sqrt{N}} \suml_{k \in \abkf{0,1}^n} {
				(-1)^{(v,k)}\ket{ k }}}.
		\end{eqnarray}
		
		Lemma (\ref{lem:abel_pres}) can be rewritten as 
		\begin{replemma}{lem:abel_pres}
			$\WalkOp$ can be presented in following way:
			\begin{equation*}
			\WalkOp = \suml_{v \in V} { \CoinOp_v \otimes \ket{ \tilde{v} } \bra{ \tilde{v} }},
			\end{equation*}
			where
			\begin{equation*}
			\CoinOp_v = 
				\abk{ \begin{array}{cccc}
				(-1)^{(v,e_1)} & & & \\
				& (-1)^{(v,e_2)} & & \\
				& & \ddots & \\
				& & & (-1)^{(v,e_m)} 
				\end{array}
				} \circ \CoinOp .
			\end{equation*}
		\end{replemma}
		
		In the sequel we use the Grover operator as a coin operator:
		\begin{eqnarray*}
			&\ket{\Psi}= \frac{1}{\sqrt{m}}\suml_{j=1}^m {\ket{j}},&  \nonumber \\
			&\CoinOp_m = 2 | \Psi \rangle \langle \Psi | - \hat{I}_m = 
			\left( \begin{array}{ccccc}
				2/n - 1 & 2/n & 2/n & \cdots & 2/n \\
				2/n & 2/n - 1 & 2/n &\cdots  & 2/n \\
				\vdots & \vdots & \ddots &  & \vdots \\
				2/n & 2/n  & 2/n &\cdots  & 2/n -1 \\
			\end{array}
			\right)&.
		\end{eqnarray*}
		
		For this operator we can find the spectrum of  $ \CoinOp_v \ $. (see \cite{mixing})

		\begin{lemma}[Coin operator spectrum] \label{lem:coinSpec}
			Let the operator $ \Gamma_d $ has the following matrix:
			\begin{eqnarray*}
				\Gamma_d = \abk{
				\begin{array}{l l l|r r r }
					2/m - 1 & 2/m & \cdots & & & \\
					2/m & 2/m - 1 & \cdots & & 2/m & \\
					\vdots & & \ddots & & & \\
					\hline
					& & & 1 - 2/m & -2/m & \cdots \\
					& -2/m & & - 2/m & 1-2/m & \cdots \\
					& & & \vdots & & \ddots \\
				\end{array}}.
			\end{eqnarray*}
			(The last $d$ rows of the Grover operator are multiplied on $ -1 $.) 
			
			If $d = 0$: $ \Gamma_0 $ has one eigenvector $\Psi = \frac{1}{\sqrt{m}}(1,\ldots, 1)^T$ corresponding to the eigenvalue $ -1 $. All other eigenvectors is orthogonal to $ \Psi $ and correspond to eigenvalue 1.
						
			If $d = m$: $ \Gamma_m $ has one eigenvector $\Psi = \frac{1}{\sqrt{m}}(1,\ldots, 1)^T$ corresponding to the eigenvalue $ 1 $. All other eigenvectors is orthogonal to $ \Psi $ and correspond to eigenvalue -1.
			
			If $ 1 \le d \le m-1 $:
			\begin{itemize}
				\item
					$ \Gamma_d $ has $ d-1 $ eigenvectors with eigenvalue $ 1 $.
				\item
					$ \Gamma_d $ has $ n - d-1 $ eigenvectors with eigenvalue $ -1 $.
				\item
					$\lambda_d=1-\frac{2 d}{m} + \frac{2i}{m}\sqrt{d(m-d)} = e ^ {i \omega_d}$ is the eigenvalue for the eigenvector
					$$
					\nu_d = \frac{1}{\sqrt{2}}(\underbrace{ \frac{-i}{\sqrt{m-d}} }_{m-d}, \underbrace{ \frac{1}{\sqrt{d}} }_{d})^T.
					$$
				\item
					$\lambda_d^*$ is an eigenvalue for eigenvector $ \nu_d^* $.
				
				The last two eigenvectors (and corresponding eigenvalues) are called \textbf{non-trivial}, the others --- trivial.
			\end{itemize}
		\end{lemma}
		
		\begin{proof}
			First two claims follows from definition of the Grover operator.
			
			The third claim:
			
			It's easy to see that the dimension of eigenspace corresponding to $ 1 $ is at least $ n-d-1 $ (in $ \Gamma_d + \hat{I} $ first $ n-d-1 $ rows are equal), due to the same reasons the dimension of eigenspace corresponding to $ -1 $ is at least $ d-1 $.
			
			It can be checked directly that $ \nu_d $ is an eigenvector with the specified eigenvalue.
			
			$ \Gamma_d $ has only real coordinates, it implies the last claim.
		\end{proof}
		
		\begin{definition}
			The sum
			\begin{equation*}
				d_{v} = \suml_{e \in S} (v,e)
			\end{equation*}
			is called the \textbf{characteristic} of vector $v$ over generating set $ S $ . (Note: this sum is taken over $ \Z $; $ d_v $ is an integer from the range $ \abks{0, m} $)
		\end{definition}
		
		It is easy to find the spectrum of $ \CoinOp_v $ using the fact that $ \CoinOp_v $ is $ \Gamma_{d_v} $ with rearranged rows.
		
		Let  $\ket{\nu_v},\ \ket{\nu_v^*} $ and $\lambda_v,\ \lambda_v^* $ denote non-trivial eigenvectors and eigenvalues of $ \CoinOp_v $.
							
	 	Hereinafter we use the symmetric initial state with starting vertex $0^n$:
		\begin{equation*}
			\ket{\psi_0} = \ket{\Psi} \ket{0^n} = \frac{1}{\sqrt{m}}\suml_{j=1}^m {\ket{j}} \ket{0^n}.
		\end{equation*}
		
		\begin{theorem}[The state decomposition] \label{th:FourierState}
			Consider the quantum walk on \Cay{\Z_2^n}{S} with symmetric initial state $ \ket{\psi_0} $ and $ \CoinOp $ is the Grover operator.
			Then $ \ket{\psi_t} $ has the form:
				\begin{equation}
					\label{eq:procState}
						\WalkOp^t \ket{ \psi_0 } = \suml_{v \in \{0,1\}^n} 
						{
							\frac{1}{\sqrt{2^n}} \left(
							a_{d_{v}} \lambda_{d_{v}}^t \ket{\nu_v} +	a_{d_{v}}^* \lambda_{d_{v}}^{*t} \ket{\nu_v^*}
							\right) \otimes	\ket{ \tilde{v} }
						},
				\end{equation}
				where 
				\begin{eqnarray*}
					& a_{d} = \frac{1}{\sqrt{2}} \abk{ \sqrt{ \frac{d}{m} } + i \sqrt{ 1 - \frac{d}{m}}}, & d \notin \abkf{0,m}; \\
					& a_{d} = \frac{1}{\sqrt{2}},& d \in \abkf{0,m}; \\
					& \nu_v = \Psi, & d_v \in \abkf{0,m}; \\
					& \lambda_0 = 1, \quad \lambda_m = -1. & 
				\end{eqnarray*}
		\end{theorem}
		\begin{proof}		
			Let's define $ a_{d} $ for $ d \notin \abkf{0,m} $ as
			\begin{equation*}
				a_{d_{v}} = \bkScal{ \nu_d }{ \Psi } = \frac{1}{\sqrt{2}} \abk{ \sqrt{ \frac{d_v}{m} } + i \sqrt{ 1 - \frac{d_v}{m}}}.
			\end{equation*}
				
			Taking into account $ |\ket{\Psi}|^2 = 1 = |a_d|^2 + |a_d^*|^2 $ we obtain
				\begin{equation}
					\ket{\Psi} = a_{d_{|v|}} \ket{ \nu_v } + a_{d_{|v|}}^* \ket{ \nu_v^* }. \label{eq:tmpSpectr}
				\end{equation}
			
			We define $ a_d, \nu_v $ for $ d_v \in \abkf{0,m} $ in accordance with the theorem's conditions. Therefore, the equation (\ref{eq:tmpSpectr}) holds for all $ v $.
			
			Taking into account \footnote{
				$ \bkScal{ 0^n }{ \widetilde {\widetilde {0^n}} } = 
		 \bra{ 0^n } \suml_{a \in \{0,1\}^n}\suml_{b \in \{0,1\}^n} \frac{1}{2^n} (-1)^{(a,b)} \ket{ b } = 
	 \bra{ 0^n } \suml_{a \in \{0,1\}^n}\suml_{b \in \{0,1\}^n} \frac{1}{2^n} \ket{ 0^n } = 1$ } 
			$ \ket{ \widetilde {\widetilde {0^n}} } = \ket{ 0^n }$ we get the decomposition of the initial state:
			\begin{equation*}
				\ket{\psi_0}=\ket{\Psi}\ket{ 0^n } = \suml_{v \in \{0,1\}^n} 
				{
					\frac{1}{\sqrt{2^n}} \abk{
					a_{d_{v}} \ket{ \nu_v } +	a_{d_{v}}^* \ket{ \nu_v^* }
					} \otimes	\ket{ \tilde{v} }
				}.
			\end{equation*}
			
			Acting on both sides with $ \WalkOp $ and using lemma \ref{lem:abel_pres} we obtain:
			\begin{equation*}
				\WalkOp^t \ket{ \psi_0 } = 
				\suml_{v \in \{0,1\}^n} \abks{ \CoinOp_v \frac{1}{\sqrt{2^n}} \abk{
									a_{d_{v}} \ket{ \nu_v } +	a_{d_{v}}^* \ket{ \nu_v^* }
									} } \otimes	\ket{ \tilde{v} } = 
			\end{equation*}
			$$
					 \suml_{v \in \{0,1\}^n} 
					{
						\frac{1}{\sqrt{2^n}} \abk{
						a_{d_{v}} \lambda_{d_{v}}^t \ket{\nu_v} +	a_{d_{v}}^* \lambda_{d_{v}}^{*t} \ket{\nu_v^*}
						} \otimes	\ket{ \tilde{v} }
					}.
			$$
		\end{proof}
		
	\section{Returning and hitting of quantum walk in \CCay{s}}
	\label{sec:retAh}
		Hereinafter the probability to hit $ \ket{1^n} $ is called the probability to hit.
				
		In the first subsection we prove the formula for calculating probabilities to hit $ 0^n $ or $ 1^n $.
		
		In the second subsection we find the probabilities to return and hit in form of sums with binomial coefficients. The sums depend on the spectrum of the walk operator, the spectrum depends on weight characteristics. We prove that the weight characteristics are Kravchuk coefficients and prove the new bound on them. Combining our results we prove our main results: Theorem \ref{th:retTime} and
Theorem \ref{th:hitTime}.
		
		\subsection{Probability to return if $ S $ is symmetric.} \label{sec:symCaseProb}
			Consider a quantum walk on \Cay{\Z_2^n}{S}. Let $ P $ be a permutation subgroup on $ n $ elements (permutations of coordinates).
			
			For each $ p \in P $ we define the operator $ \hat{p} $ on the main space by the rule
			\begin{equation*}
				\hat{p} \ket{v} = \hat{p} \ket{v_1,v_2,\ldots,v_n} = \ket{v_{p(1)},v_{p(2)},\ldots,v_{p(n)}}=\ket{p(v)}.
			\end{equation*}
			Note that $ \hat{p} $ is a unitary operator on the main space.
			
			For each $ p \in P $ we define the operator $ \hat{p}' $ on the coin space: let the $ \hat{p}' $ take each $ e \in S $ such that $ p(e) \in S $ to $ \ket{p(e)} $ and each $ e \in S$ such that $ p(e) \notin S $ to 0.
			
			Note that $ \hat{p'} $ is an unitary operator (a permutation of basis states) on the coin space iff $ \forall e \in S \ p(e) \in S$. 
			
			\begin{theorem} \label{th:decompInSym}
				Let $ P $ acts transitively on $ S $. Then for the probability to return and the probability to hit we have:
				\begin{equation*}
					(\hat{I}  \otimes \ket{0^n}\bra{0^n})  \WalkOp^t \ket{\psi_0} = \bkScal{\psi_0}{\psi_t}\ket{\psi_0},
				\end{equation*}
				\begin{equation*}
					(\hat{I}  \otimes \ket{1^n}\bra{1^n}) \ket{\psi_t} = \bkScal{\Psi, 1^n}{\psi_t}\ket{\psi_0}.
				\end{equation*}
				(Here $ \ket{\psi_0} = \ket{\Psi}\ket{0^n} $.)
			\end{theorem}
			
			\begin{proof}
				We prove the first formula, the second proved similarly.
				
				The operator $ \hat{p}' \otimes \hat{p} $ commutes with $ \ShiftOp $ (can be easily checked on basis vectors ). The state $\ket{\Psi}$ is an eigenvector of $ \hat{p'} $. So, $ \hat{p}' \otimes \hat{p}$ commutes with $ \CoinOp = 2 \ket{\Psi}\bra{\Psi} - \hat{I} $ . That means that $ \hat{p'} \otimes \hat{p}$ commutes with $ \WalkOp $.
				
				Taking into account $\hat{p} \ket{0^n} = \ket{0^n}$ we get:
				\begin{equation*}
					(\hat{p'}  \otimes \hat{I}) (\hat{I}  \otimes \ket{0}\bra{0}) = (\hat{I}  \otimes \ket{0}\bra{0})(\hat{p'} \otimes \hat{p}).
				\end{equation*}
				
				Then
				\begin{eqnarray*}
					& (\hat{p'}  \otimes \hat{I}) (\hat{I}  \otimes \ket{0}\bra{0}) \WalkOp^t \ket{\psi_0} = 
					(\hat{I}  \otimes \ket{0}\bra{0}) (\hat{p'}  \otimes \hat{p})  \WalkOp^t \ket{\psi_0} = & \\
					&(\hat{I}  \otimes \ket{0}\bra{0}) \WalkOp^t  (\hat{p'}  \otimes \hat{p}) \ket{\psi_0} = 
					(\hat{I}  \otimes \ket{0}\bra{0})  \WalkOp^t \ket{\psi_0}.&
				\end{eqnarray*}
				
				The last result and transitivity of $ P $ implies that all amplitudes of $$ (\hat{I}  \otimes \ket{0}\bra{0})  \hat{Q}^t \ket{\psi_0}$$ in the expansion in the standard basis are equal:
				
				\begin{eqnarray*}
					(\hat{I}  \otimes \ket{0}\bra{0})  \hat{Q}^t \ket{\psi_0} = \alpha \ket{\psi_0},
				\end{eqnarray*}
				
				\begin{eqnarray*}
					(\hat{I}  \otimes \ket{0}\bra{0})  \hat{Q}^t \ket{\psi_0} = ( \bra{\psi_0} \hat{Q}^t \ket{\psi_0} ) \ket{\psi_0}.
				\end{eqnarray*}
			\end{proof}
			
			Note that in case of \CCay{s} the theorem holds if $ P $ is the set of all permutations on $ n $ elements.
			
			\begin{lemma}
				$ \forall t \in \Z $, $ \forall x \in \Z_2^n $ such that $ |x| \notin \abkf{0,n} $ the probability to hit vertex $ x $ is less or equal to $ \frac{1}{n} $
			\end{lemma}
			\begin{proof}
				Let $ l = |x| $ and 
				\begin{equation*}
					\hat{\Pi} = \suml_{y:\ |y|=l} \ket{y}\bra{y} .
				\end{equation*}
				
				Repeating the assignments from the proof of Theorem \ref{th:decompInSym} we have:
				\begin{equation*}
					(\hat{I}  \otimes \Pi)  \hat{Q}^t \ket{\psi_0} = \alpha \ket{\Psi} \abks{\suml_{y:\ |y|=l} \ket{y} }.
				\end{equation*}
				
				Therefore, for all $ x,y$ such that $|x|=|y| $ the probability to hit $ x $ is equal to the probability to hit $ y $.
				
				Number of vertexes with weight $ l $ is at least $ n $.
			\end{proof}
		
		\subsection{Weight characteristics and linear codes }
			\begin{definition}[Linear code]
			The linear $[m,n]$ (of length $m$, rank $n$) code over a field $F$ is a linear subspace of dimension $n$ in linear space $F^m$. If $A \in M_{m \times n}$ is the matrix of a linear operator $\hat{A}: F^n \to F^m$, the matrix $A$ is called the generating matrix for the code $\operatorname{Im} A$.
			\end{definition}		
			
			\begin{definition}
				The\textbf{ code of Cayley graph} $(G,S)$, where $G = F^n$ is a code with a generating matrix $ A $ such that the $ i $-th row of $ A $ is the $ i $-th generating element.
			\end{definition}
			
			Further we work with the field $\Z_2$.
			
			In this case the generating matrix can be defined in the following way:
			\begin{eqnarray*}
				A x = ( (e_1, x), \ldots, (e_{|S|},x))^T.
			\end{eqnarray*}
			
			\newcommand{\rk}{\operatornamewithlimits{rk}}
			\begin{remark}
				The code for $\cay{\Z_2^n}{S}$ is $ [m, \rk S] $ linear (where $ m = |S| $).
			\end{remark}
			
			\begin{remark}
				In terms of linear codes, \textbf{vector characteristic} $ d_{v} $ (defined in Section \ref{sec:qwZ2} ) is $ |Av| $.
			\end{remark}
			
			\begin{definition}
				The \textbf{weight function} of $[m,n]$ code is $W(x) = \suml_{k=0}^m W_k x^k$, where $W_k$ is the number of code words that has Hamming weight $ k $. The numbers $W_k$ are called \textbf{weight coefficients}.
			\end{definition}
			
			\begin{remark}
				According to Theorem \ref{th:FourierState} and Lemma \ref{lem:coinSpec}, we get that weight coefficients are determine spectrum of the quantum walk on \Cay{\Z_2^n}{S}.
			\end{remark}

			In the case of \CCay{s}, if $ |v_1| = |v_2| $ then $ |A v_1| = |A v_2| $. As it was defined previously $ d_{v} = |A_v| $. So, in \CCay{s} $ d_{v} $ depends only on $ |v| $. 
			
			\begin{definition}
				We say that $ d_v $ is the \textbf{weight characteristic } for weight $ k = |v| $ and write $ d_k^s $. Also we define corresponding $ \lambda_k^s, \omega_k^s $: $$\lambda_{d_k^s}=1-\frac{2 d_k^s}{m} + \frac{2i}{m}\sqrt{d_k^s(m-d_k^s)} = e ^ {i \omega_k^s}.$$
			\end{definition}
			
			\begin{lemma} \label{eq:sumMainVie}
				~
				
				In \CCay{s}:
				\begin{itemize}
				\item
				the probability to return at moment $ t $ is
				\begin{equation} \label{eq:retMainSum}
					\sqrt{Pr} = \abkm{ \suml_{k = 0}^n \frac{1}{2^n} \binom{n}{k} \cos{ (\omega_k^s  t)} },
				\end{equation}
				\item
				the probability to hit at moment $ t $ is
				\begin{equation} \label{eq:hitMainSum}
					\sqrt{Pr} = \abkm{ \suml_{k = 0}^n \frac{(-1)^k}{2^n} \binom{n}{k}  \cos{ (\omega_k^s  t)} }.
				\end{equation}
				\end{itemize}
			\end{lemma}
			\begin{proof}
				Combining Theorem \ref{th:FourierState} and Theorem \ref{th:decompInSym} we get
				\begin{eqnarray}
					&\bra{ \psi_0 } \WalkOp^t \ket{ \psi_0 } = 
					\suml_{v \in \{0,1\}^n} 
					{
						\frac{1}{\sqrt{2^n}} \left(
						a_{d_{v}} \lambda_{d_{v}}^t  \bkScal{ \Psi }{\nu_v } +	a_{d_{v}}^* \lambda_{d_{v}}^{*t}\bkScal{ \Psi  }{\nu_v^* }
						\right)  \bkScal{ 0^n }{  \tilde{v} }
					} =& \nonumber\\\nonumber
					&\suml_{v \in \{0,1\}^n} 
					{
						\frac{1}{\sqrt{2^n}} \left(
						a_{d_{v}} \lambda_{d_{v}}^t  a_{d_{v}}^* +	a_{d_{v}}^* \lambda_{d_{v}}^{*t} a_{d_{v}}
						\right)  \frac{1}{\sqrt{2^n}}
					}=&\\\nonumber
					&\suml_{v \in \{0,1\}^n}
					{
						\frac{1}{2^n} \left(
							\frac{ \lambda_{d_{v}}^t}{2} +  \frac{ \lambda_{d_{v}}^{*t}}{2}
						\right)
					}=
					\suml_{v \in \{0,1\}^n}
					{
						\frac{1}{2^n} 
							\cos{ (\omega_{v}  t)}
					} = &\\ &\suml_{k=1}^n
					{
						\frac{1}{2^n} \binom{n}{k}
							\cos{ (\omega_k^s  t)}
					}.\label{eq:retEst} &
				\end{eqnarray}
				
				Taking into account that $ \bkScal{1^n}{\tilde{v}} = (-1)^{|v|} $ the second claim can be proved in the same way.
			\end{proof}
			
			\begin{lemma}
				In \CCay{s} we have:
				\begin{equation} \label{eq:weight}
					d_{\vec{v}} = d_{|v|}^s = \suml_{l=0}^{ \lfloor \frac{s-1}{2} \rfloor } { \binom{|v|}{2l+1} \cdot \binom{n-|v|}{s-2l-1} }.
				\end{equation}
			\end{lemma}
			\begin{proof}
				$ d_{|v|}^s $ is the number of ways to place $ s $ ones in $ n $ positions such that there are an odd number of ones is placed in $ |v| $ predetermined positions.
			\end{proof}
			
		\subsection{Kravchuk coefficients estimation} \label{sec:Krevchuk}
			In this section we give a bound on 
			\begin{equation*}
				\cos \omega_k^s = 1 - \frac{2 d_k^s}{m},
			\end{equation*}
			where
			\begin{eqnarray*}
				& m = \binom{n}{s} &\\
				& d_k^s = \suml_{l\ odd; l \in [0,s] } \binom{k}{l} \binom{n-k}{s-l}, &\\
			 	\nonumber \\
				& \abkm{k-\frac{n}{2}} \le \frac{n}{2}\delta ,& \\
				& \delta = \sqrt{\frac{2f(n)}{n}}. &
			\end{eqnarray*}
			and $ f $ is sufficiently small ($ f = o(n) $).
			
			\begin{eqnarray} \nonumber
				&\cos \omega_k = 1 - \frac{2 d_k}{m} = 
					1-\frac{ 2\suml_{l\ odd; l \in [0,s] }^s \binom{k}{l} \binom{n-k}{s-l} }
					{\binom{n}{s}} = &\\\nonumber
				& \frac{\suml_{l \in [0,s]} \binom{k}{l} \binom{n-k}{s-l} -2\suml_{l\ odd; l \in [0,s] }
					\binom{k}{l} \binom{n-k}{s-l} }
					{\binom{n}{s}} = & \\\nonumber
				& \frac{\suml_{l\ even; l \in [0,s]} \binom{k}{l} \binom{n-k}{s-l} -
					\suml_{l\ odd; l \in [0,s] } \binom{k}{l} \binom{n-k}{s-l} }
					{\binom{n}{s}} = &\\
				& \frac{\suml_{l=0}^s (-1)^l\binom{k}{l} \binom{n-k}{s-l}}
					{\binom{n}{s}} \label{eq:cosBin}. &
			\end{eqnarray}
			
			The numerator of the fraction (\ref{eq:cosBin}) is known as Kravchuk coefficient \cite{kravchuk}:
			\begin{equation*}
				\phi_{k,n}(s) = \suml_{l=0}^s (-1)^l\binom{k}{l} \binom{n-k}{s-l}.
			\end{equation*}
			
			\begin{remark}
				$f=(1-x)^{k}(1+x)^{n-k}$ is the generating function of $ \phi_{k,n}(s)$ \cite{kravchuk}.
			\end{remark}
			
			\begin{coroll} \label{cl:cosEst}
				Let $ 2 \del n, \ k = n/2 $. If $ 2 \del s $ then $$ \phi_{k,n}(s) = (-1)^{s/2} \binom{n/s}{s/2}. $$ If $ 2 \ndel s $ then $ \phi_{k,n}(s) = 0 $.
			\end{coroll}
			\begin{proof}
				\begin{eqnarray*}
					\suml_{s=0}^{n} \phi_{k,n}(s) x^s = (1-x)^{n/2}(1+x)^{n/2} =\\ (1-x^2)^{n/2} = 
					\suml_{t=0}^{n/2} (-1)^t \binom{n/2}{t} x^{2t}.
				\end{eqnarray*}				
			\end{proof}
			
			\begin{theorem}
				\label{th:weightEst}
					For $s^2 = o(n), s < \delta n $ we have:
				\begin{eqnarray} \label{eq:weightEst}
					\abkm{\cos \omega_k^s} = O\abk{ (s+1)!  \delta^{s/2}  }.
				\end{eqnarray}
			\end{theorem}
			\begin{proof}
				
				According to $d_k^s = d_{n-k}^s$ we assume that $k \le n/2$ .
				
				Using $ s^2=o(n) $ we get $ \abk{1 - s/n}^s \to 1, n \to +\infty $
				
				Consider a Boolean random vector with $n$ coordinates that have exactly $s$ coordinates equal to 1. All outcomes have the same probability. We divide coordinates on two parts: first one has length $2k$, the other --- $n-2k$.

				Let's define events: 
				\begin{itemize}
					\item $A_n^s=\abkf{ \mbox{in first } k \mbox{ coordinates number of ones is even} }$
					\item $B_n^s=\abkf{ \mbox{in first } k \mbox{coordinates number of ones is odd} }$
					\item $C_l = $ \{ in coordinates with numbers from range $[2k+1, n]$ number of ones is $s-l$ \}
				\end{itemize}
	
				We denote the indicator function as $$\chi[\mbox{ conditiona A}] = \left\lbrace \begin{array}{cc}
									1,&\mbox{A is true} ;\\
									0,&\mbox{A is false}.
								\end{array} \right. $$
				
				Note that:
				\begin{itemize}
					\item $\cos \omega_k^s = P(A_n^s)-P(B_n^s),$
					\item $P(\bigsqcup\limits_{l=0}^s C_l )= 1,$
					\item $P(C_l)= \chi[ s-l \le n-2k] \binom{2k}{l} \cdot \binom{n-2k}{s-l} / \binom{n}{s} ,$
					\item $P(A_n^s|C_l) = P(A_{2k}^{l}),$
					\item $P(B_n^s|C_l) = P(B_{2k}^{l}).$
					\item Due to Corollary \ref{cl:cosEst}:
	
						 $ \binom{2k}{l} \abk{P(A_{2k}^{l})-P(B_{2k}^{l})} = (-1)^{l/2} \chi[2 \del l] \binom{k}{l/2}. $
				\end{itemize}
				
				Combining these relations we get:
				\begin{eqnarray}
					& \abkm{P(A_n^s)-P(B_n^s)} = \label{eq:reallKrav}
						\abkm{\suml_{l=0}^s P(A_n^s|C_l) P(C_l)- \suml_{l=0}^s P(B_n^s|C_l) P(C_l)} = & \\ \nonumber
					& \abkm{\suml_{l=0}^s \abks{P(A_n^s|C_l) - P(B_n^s|C_l)} P(C_l)} \le & \\ \nonumber
						& \suml_{l=0}^s \abkm{P(A_n^s|C_l)-P(B_n^s|C_l)}P(C_l) = & \\
						\nonumber \\  \label{eq:weigthEstBefApp}
					& =\suml_{l=0}^s \abkm{ P(A_{2k}^{l})-P(B_{2k}^l) }P(C_l).  &
				\end{eqnarray}
					The bound (\ref{eq:weightEst}) follows from (\ref{eq:weigthEstBefApp}) by direct computations. The end of the proof is given in Appendix \ref{sec:kravApp}.
			\end{proof}
			
			\begin{lemma}
				In conditions of Theorem \ref{th:weightEst} we have:
				\begin{equation*}
					\frac{\pi}{2} -\omega_k^s = 1 - \frac{2 d_k^s}{m} + O\abk{ (s+1)!^3 \delta^{3s} }.
				\end{equation*}
			\end{lemma}
			\begin{proof}
				Let $ \alpha = \frac{\pi}{2} - \omega_k^s $ ($ \alpha \in \abks{-\frac{\pi}{2}, \frac{\pi}{2}} $). Then
				\begin{equation*}
					\sin \alpha = \cos \omega_k^s = 1 - \frac{2d_k^s}{m}
				\end{equation*}
				Acting $ \arcsin $ on both sides of the equation completes the proof (we use that~$\arcsin x = x + O(x^3) $).
			\end{proof}			
			
			\begin{coroll}
				If $ t = \frac{\pi}{2} m + \epsilon, \ t \equiv m \ \pmod{2} $
				\begin{eqnarray}
				  \nonumber
	  					\cos \abk{ \omega_k^s t } = (-1)^{\frac{t-m}{2} + d_k^s} \Big( 1 - O(\epsilon^2 \delta^{2s}) - &  &  \\
	  					 - 
	  					 O(m^2 \abks{(s+1)!}^6 \delta^{6s} ) & -  O(\epsilon m \abks{(s+1)!}^3\ \delta^{4s}) \Big).& 
				  \label{eq:estPi2}
				\end{eqnarray}
			\end{coroll}
			\begin{proof}
				\begin{eqnarray*}
					\cos \omega_{d_k^s} t = \cos (\frac{\pi}{2} - (\frac{\pi}{2} - \omega_k^s) ) t =
				\end{eqnarray*}
				\begin{eqnarray*}
					\cos \abk{ \frac{\pi}{2}t - (1 - \frac{2 d_k^s}{m})( \frac{\pi}{2} m + \epsilon ) + O(m \abks{(s+1)!}^3 \delta^{3s} )   } =\\
					\cos \abk{ \frac{\pi}{2}(t - m) - {\pi} d_k^s + O(\epsilon \delta^s) + O(m \abks{(s+1)!}^3\ \delta^{3s} )   } = \\
					(-1)^{\frac{t-m}{2}} (-1)^{d_k^s} \cos \abk{ O(\epsilon \delta^s) + O(m \abks{(s+1)!}^3\ \delta^{3s} )  } = \\
					(-1)^{\frac{t-m}{2}+d_k^s} ( 1 - O(\epsilon^2 \delta^{2s}) -O(m^2 \abks{(s+1)!}^6 \delta^{6s} ) - \nonumber \\  -  O(\epsilon m \abks{(s+1)!}^3\ \delta^{4s} ) ).
				\end{eqnarray*}
			\end{proof}
			
			The proof of the next corollary is similar.
			\begin{coroll}
				If $ t = \frac{\pi} m + \epsilon, \ t \equiv m \ \pmod{2} $
				\begin{eqnarray}
					\nonumber
					\cos \abk{ \omega_k^s t } = (-1)^{\frac{t-m}{2}} \Big( 1 - O(\epsilon^2 \delta^{2s})- &  &  \\
					 - 
					O(m^2 \abks{(s+1)!}^6 \delta^{6s} )&  -  O(\epsilon m \abks{(s+1)!}^3\ \delta^{4s}) \Big).& 
					\label{eq:estPi}
				\end{eqnarray}
			\end{coroll}
			
			\begin{lemma}
				The following holds:
				$$
					d_{k+2}^s - d_{k}^s \equiv 0 \ \pmod{2}.
				$$
			\end{lemma}
			\begin{proof}
				
				\begin{equation*}\begin{array}{c}
					f_{k+2,n}-f_{k,n}= (1-x)^{k+2} (1+x)^{n-k-2} - (1-x)^k (1+x)^{n-k} = \\
					f_{k,n-2} \abks{(1-x)^2 - (1+x)^2 } = (-4x) f_{k,n-2} 
				\end{array}\end{equation*}
				
				Therefore, we get:
				\begin{equation*}\begin{array}{c}
					\phi_{k+2,n}(s)-\phi_{k,n}(s) \equiv 0 \pmod{4}
				\end{array}\end{equation*}
				
				The equation
				\begin{equation*}\begin{array}{c}
					m-2 d_k^s = \phi_{n,k}(s)
				\end{array}\end{equation*}
				completes the proof.
			\end{proof}
			
			\begin{coroll}
				~\label{cl:evOd}
				
				If $ k $ is even then $ d_k^s $ is even.
				
				If $ k $ is odd, then holds $ d_k^s \equiv \frac{m s}{n} \pmod{2} $.
			\end{coroll}
			\begin{proof}~
			
				$ d_0^s = 0 $
				
				$ d_1^s = \binom{n-1}{s-1} = \frac{m s}{n} $
			\end{proof}

		\subsection{Return and hitting time } \label{sec:retTime}
			
			\begin{theorem} \label{th:retTime}
				Consider a quantum walk in \CCay{s}.
				
				If the following conditions
				\begin{itemize}
					\item $ s! \le n^{s/8} $;
					\item $ t = \pi m + \epsilon$;
					\item $ \epsilon = n^{\beta s}$;
					\item $2 \del (t-m)$;
					\item $1/4 < \beta < 1/2$;
				\end{itemize}
				hold then the probability to return is:
				\begin{eqnarray}
					Pr = 1 - O(\frac{1}{n}) - O(\frac{\ln^s n}{n^{s-2\beta s} }) \label{eq:mainRetEst}.
				\end{eqnarray}
			\end{theorem}

			\begin{proof}
				To prove the theorem we estimate the sum
				$$
					\sqrt{Pr} = \abkm{ \suml_{k = 0}^n \frac{1}{2^n} \binom{n}{k} \cos{ (\omega_k^s  t)} } .  \eqno{(\ref{eq:retMainSum})}
				$$
				
				Let $J = [\frac{n}{2}(1-\delta),\frac{n}{2}(1+\delta)]$ where $ \delta = \sqrt{\frac{2 \ln n}{n}} $ .
				
				Using Chernoff bound we get:
				\begin{equation*}
					\abkm{
					\suml_{k \notin J}
		 \frac{1}{2^n} \binom{n}{k}
							\cos{ (\omega_{d_k^s}  t)} 
					} \le 
					\suml_{k \notin J}
	 \frac{1}{2^n} \binom{n}{k}
					\le 2 e^{-\frac{\delta^2 n}{2}} = O(\frac{1}{n}). \label{eq:chEstRet}
				\end{equation*}.
				
				Using (\ref{eq:estPi}) we obtain:
				\begin{eqnarray*}
					& \suml_{k \in J}
					 	\frac{C_n^k}{2^n}
						\cos{ (\omega_{d_k^s}  t)} = &
				\end{eqnarray*}
				\begin{eqnarray*}
					\suml_{k \in J}
						\frac{C_n^k}{2^n}
							(-1)^{\frac{t-m}{2}} (-1)^{\frac{t-m}{2}} \Big( 1 - O(\epsilon^2 \delta^{2s}) - &  & \\
							- O(m^2 \abks{(s+1)!}^6 \delta^{6s} )  & -  O(\epsilon m \abks{(s+1)!}^3\ \delta^{4s}) \Big)& = \nonumber
				\end{eqnarray*}
				\begin{eqnarray} \label{tmp_eq1}
					\abks{ 1 - O(\frac{1}{n})}
							(-1)^{\frac{t-m}{2}} (-1)^{\frac{t-m}{2}} \Big( 1 - O(\epsilon^2 \delta^{2s}) - &  & \\
							- O(m^2 \abks{(s+1)!}^6 \delta^{6s} )  & -  O(\epsilon m \abks{(s+1)!}^3\ \delta^{4s}) \Big)& \nonumber
				\end{eqnarray}

				The first condition of the theorem implies $s^b = o(n) \ \forall b >0$ .
				
				So, we get:
				\begin{eqnarray}
					{\frac{s^2 \abks{(s+1)!}^4 \ln^{2s} n }{ n^{s} }}  =
					O\abk{\frac{\ln^{2s} n} {n^{s/2 -1}}}, \label{eq:rett1} \\
					{\frac{s\abks{(s+1)!}^2 \ln^{2s} n}{n^{s-\beta s}}} = 
					O\abk{ \frac{\ln^{2s} n}{n^{\frac{3}{4}s -\beta s - 1}} } .\label{eq:rett2}
				\end{eqnarray}
				
				Also we can estimate $ m $ as
				\begin{eqnarray*}
					m \le \frac{n^s}{s!}.
				\end{eqnarray*}
				
				Combining (\ref{tmp_eq1}), (\ref{eq:rett1}), (\ref{eq:rett2}) completes the proof:
				\begin{eqnarray*}
					& 1 \ge \sqrt{Pr} \ge \abkm{\suml_{k \in J} \frac{C_n^k}{2^n} \cos{ (\omega_{d_k^s}  t)}} - 
											\abkm{\suml_{k \notin J} \frac{C_n^k}{2^n}\cos{ (\omega_{d_k^s}  t)}} & = \\
					& 1 - O(\frac{\ln^s n}{n^{s-2\beta} }) - O(\frac{1}{n}).&	\nonumber
				\end{eqnarray*}
			\end{proof}
			
			\begin{remark}
				Replace the last condition $ 1/4 \le \beta \le 1/2 $ by $ \beta \le 1/4 $. Then the following bound:
				\begin{equation*}
		 			1 - O(\frac{1}{n}) - O(\frac{\ln^{2s} n}{n^{s/2-1/2} }).
				\end{equation*}
				can be proved in a similar way.
			\end{remark}
			
			\begin{theorem} \label{th:retTime2}
				Consider a quantum walk in \CCay{s}.
				
				If the following conditions
				\begin{itemize}
					\item $ \frac{ms}{n} \equiv 0 \pmod{2}$ 
					\item $ s! \le n^{s/8} $;
					\item $ t = \frac{\pi}{2} m + \epsilon$;
					\item $ \epsilon = n^{\beta s}$;
					\item $2 \del (t-m)$;
					\item $1/4 < \beta < 1/2$;
				\end{itemize}
				hold then the probability to return is:
				\begin{eqnarray}
					Pr = 1 - O(\frac{1}{n}) - O(\frac{\ln^s n}{n^{s-2\beta s} }) \label{eq:mainRetEst2}.
				\end{eqnarray}
			\end{theorem}
			
			\begin{proof}
				The proof is almost similar to the proof of the previous theorem. It's sufficient to use (\ref{eq:estPi2}) instead of (\ref{eq:estPi}) and recall Corollary \ref{cl:evOd}.
			\end{proof}
			
			\begin{coroll} \label{th:hitTime2}
				In conditions of Theorem \ref{th:retTime2} the probability to heat is $ o(1) $.
			\end{coroll}
			
			\begin{theorem} \label{th:hitTime}
				Consider a quantum walk in \CCay{s}.
				
				If the following conditions
				\begin{itemize}
					\item $ \frac{ms}{n} \equiv 1 \pmod{2}$ 
					\item $ s! \le n^{s/8} $;
					\item $ t = \frac{\pi}{2} m + \epsilon$;
					\item $ \epsilon = n^{\beta s}$;
					\item $2 \del (t-m)$;
					\item $1/4 < \beta < 1/2$;
				\end{itemize}
				hold then the probability to return is:
				\begin{eqnarray}
					Pr = 1 - O(\frac{1}{n}) - O(\frac{\ln^s n}{n^{s-2\beta s} }) \label{eq:mainHitEst}.
				\end{eqnarray}
			\end{theorem}
		
			\begin{remark}
				The result for $ s=1 $ is proved by Julia Kempe in \cite{kempe}.
			\end{remark}			
			
	\section{Antipodality in \CCay{s}} \label{sec:antipod}
		~
		
		\begin{definition}
			Let  $ G $ be a graph.
			
			A vertex $ u $ is an \textbf{antipodal} to a vertex $ v \ne u $ if for any automorphisms $ f \in \operatorname{Aut} G$ such that $  f(v) = v  $ holds $ f(u) = u $.
		\end{definition}
		
		\begin{remark}
			The antipodality relation is transitive but generally is not symmetric.
		\end{remark}

		\begin{definition}
			We define a layer of weight $ k $ as 
			\begin{equation*}
				L_k = \abkf{v \in G(s) : |v| = k}
			\end{equation*}
		\end{definition}
		
		\begin{definition}
			We say that layers $ t $ and $ l $ are connected if $\exists v,u \ |v| = l, |u| = t:\ (u,v) \in E$ and write $ (t,l) \in E' $.
		\end{definition}
		
		Below we show that for the graphs \CCay{s} with sufficiently small $ s $ ($ s \le n/6 $ is sufficient) each vertex $ v  $ has exactly one antipodal vertex: $ v + 1^n $.
	
		To prove the claim we show that if $ 0^n $ is a fixed point then layers of $ G(s) $ are invariant (Theorem \ref{th:inv}).
				
		In the paper \cite{krasin} Krasin proved the theorem if $ s $ is odd and \linebreak $ s \notin \abkf{\frac{n-1}{2}, \frac{n+1}{2}, n/2} $.
		
		We use similar technique to generalize his result on the even case.
		
		\begin{definition}
			Let $ G_2(s) $ denote the sub graph of \CCay{s} inducted by all vertices with even weight. Then the graph $ G(s) $ is defined as follows
			\begin{equation*}
			 G(s) = \left\lbrace \begin{array}{c}
			 	\ccay{s}, \ if \ s \equiv 1 \pmod{2} ,\\
			 	G_2(s),  \ if \ s \equiv 0 \pmod{2}.
			 \end{array} \right.
			\end{equation*}
		\end{definition}
		
		The Theorem \ref{th:conCom} states that $ G(s) $ is exactly the connected component of \CCay{s} that contains $ 0^n $.

		\subsection{Structure of $ G(s) $}
		
			\begin{theorem} \label{th:conCom}
				Let $s < n$, then
				\begin{enumerate}
					\item
					$s$ is even  $ \Rightarrow $ \CCay{s} has 2 connected components: vertexes of even and odd weights.
					\item
					$s$ is odd $ \Rightarrow $ \CCay{s} is a connected graph.
				\end{enumerate}
			\end{theorem}
	
			\begin{proof}
				Let denote the linear closure of $ S $ (over  $\Z_2^n$) as $Lin(S)$.
				
				It is easy to see that the set of vertexes that can be reached from $ v $ is exactly $v + Lin(S) = \abkf{v+e|e \in Lin(s)}$.
				
				\begin{enumerate}	
					\item
					All elements of $Lin(S)$ have even weight.
					
					Let's show that $Lin(S)$ contains each vector of weight 2: for all $m \neq k$ (coordinate numbers) let's choose $e_1 \in S$ having $k$--th coordinate equal to one, $m$--th --- zero; the $e_2$ is obtained from $e_1$ by inverting $m$--th and $k$--th coordinates. Then $e_1 + e_2 \in Lin(S)$ is a vector having ones only in $m$--th and $k$--th coordinates.
					
					Therefore $Lin(S)$ contains all vectors with even weight, so if two vectors difference is even, they are in the same connected component.
					
					Due to the equivalence $$ |v+u| \equiv |v| + |u| \pmod{2}, $$ vertexes with even and odd weights can't be connected.
					
					\item
					Similarly the previous, $Lin(S)$ contains all vectors with even weight. Since $Lin(S)$ has a vector with odd weight it contain all vectors of weight 1. These vectors generate the whole space $\Z_2^n$ (they are basis).
				\end{enumerate}
			\end{proof}
			
			\begin{theorem}[Connections between layers] \label{th:conLay}
				~
				
				The neighbours set of level $ l $ is 
				\begin{equation*}
					N(l) = \abkf{t \in \Big[ |l-s|, min(l+s, n - (l + s - n) ) \Big]\ \Big|\ t \equiv |l-s| \pmod{2}}.
				\end{equation*}
			\end{theorem}
			\begin{proof}				
				Let $v$ has weight $l$ and $u$ is a neighbour of $v$, which means $ \exists e, |e| =s: u = v + e$.				
				
				We denote $ (v_1 \cdot u_1, \ldots, v_n \cdot u_n) $ as $ v \cap u $.
				
				Let $p = |e \cap v|$ (the number of coordinates that equal to 1 in both $ e $ and $ v $). It's clear that $ p $ is at least $l+s-n$:
				\begin{equation} \label{eq:conLayTmp1}
				 \max(0, l+s-n) \le p \le \min(|v|, |e|) = \min(l, s).
				\end{equation}
				
				Using that
				\begin{itemize}
					\item
					$l-p$ coordinates are one in $ v $ and zero in $ e $,
					\item
					$s-p$ coordinates are one in $ e $ and zero in $ v $,
				\end{itemize}
				
				we get:
				\begin{equation}
					|u| = (l-p) + (s-p) = (l+s) - 2p. \label{weightEq}
				\end{equation}
				
				Combining (\ref{eq:conLayTmp1}) and (\ref{weightEq}), we get:
				\begin{equation*} 
					(l+s) - 2 \min(l,s) \le |u| \le l+s - 2 \max(0, l+s-n).
				\end{equation*}
				Using the equality $(l+s) - 2 \min(l,s) = |l-s|$ we obtain
				$$
					N(l) \subset \abkf{t \in [ |l-s|, \min(l+s, n - (l + s - n) ) ]| t \equiv |l-s| \pmod{2}}
				$$
				
				To prove inclusion in the other side we just note that the bound (\ref{eq:conLayTmp1}) is exact (for all $ v $).
			\end{proof}	
				
			\begin{coroll}[Number of connections]
				\label{joins_num}
				Let $ (l,t) \in E' $. Then each vertex $v, \ |v| = l$ has
				\begin{equation*}
					\binom{l}{(s + l)/2 - t/2} \binom{n-l}{(s+t)/2-l/2}
				\end{equation*}
				neighbours of weight $ t $ in $ G(s) $.
			\end{coroll}
			\begin{proof}						
							
				A vector $e$ correspond to a neighbour of weight $t,\ p := |e \cap v| $ $\Leftrightarrow$
				
				$$ (l-p) + (s-p) = t . $$ 
				
				We get $ p = (s + l - t) / 2 $.
				
				The number of $e$, such that $|e| = s, \ |e \cap v| = p$ is
				
				\begin{equation*}
				\binom{l}{p} \binom{n-l}{s-p} = \binom{l}{(s + l)/2 - t/2} \binom{n-l}{(s+t)/2-l/2}
				\end{equation*}
			
			\end{proof}
				
		\subsection{The invariance property of layers}
			
			\begin{lemma}\label{lm:connum}
				Put $$ k(l) = \binom{ l }{ l/2} \binom{n-l}{s-l/2} $$ for $ 0 \le l \le 2(s-1)  $.
				Let $ 6s \le n $.
				
				Then $ k(l+2) < k(l) $.
			\end{lemma}
			\begin{proof}
				By definition:
				\begin{eqnarray*}
					 k(l+2) = k(l) \cdot \frac{ (l+1) (l+2) (s-l/2) (n-l-s+l/2)  }{ (l/2 + 1) ^ 2  (n-l-1)(n-l)}.
				\end{eqnarray*}
				The lemma follows from the following bound:
				\begin{eqnarray*}
					 &\frac{ (l+1) (l+2) (s-l/2) (n-l-s+l/2)  }{ (l/2 + 1) ^ 2  (n-l-1)(n-l)}=&  \\
					 &\frac{ 4 (l+1) (s-l/2) (n-l-s+l/2)  }{ (l + 2)  (n-l-1)(n-l)} \le 
					 4 \frac{ l+1 } {l+2} \frac{ s (n-s)  }{ (n-2s+1)(n-2s+2) } <& \\
					 & 4 \cdot \frac{  (n/s-1)  }{ (n/s-2+1/s)(n/s-2+2/s) } \le &\\&  4 \frac{  (n/s-2)  }{ (n/s-2)(n/s-2) }
					 \le 1.&
				\end{eqnarray*}
			\end{proof}
			
			\begin{theorem}\label{th:inv}
				Let $ 6s \le n $, $ f \in \operatorname{Aut} G(s) $.
				
				If $ 0^n $ is a fixed point of $ f $ then each layer is invariant under action of $ f $.
			\end{theorem}			
			\begin{proof}
				In the case of odd $ s $, it's sufficient to show that $ L_1 $ is invariant. In the case of even $ s $ it's sufficient to show that $ L_2 $ is invariant.
				
				If $ 0^n $ is a fixed point then $ L_s $ is an invariant.
				
				Let $ k(v) $ denote the number of vertices $ u $, such that $ (v,u) \in E $ and $ |u|=s $.
				
				Since $ L_s $ and the adjacency matrix are invariants then $ k(v)=k(f(v)) $. The function $ k(v) $ depends only from weight of $ v $. We define the function $ k(l) $ by rule: $ k(l)=k(v) $ iff $ |v|=l $.
				
				For each $ l \notin [0,2s] \quad k(l) = 0 $. According to Lemma \ref{lm:connum} ,if $ l, l' \in [0,2s] $ and $ l,l' $ are even, then  $\ k(l) \ne k(l') $.
				
				Therefore, if $ v \in L_l $, $ l \in [0,2s] $, $ l $ is even then $ f(v) \in L_l $. It means that $ L_l $ is invariant.	
				
				Let $ s $ is even. According to above results, $ L_2 $ is an invariant.
				
				Let $ s $ is odd. According to above results, $ L_{s-1} $ and $ L_{s+1} $ are invariants. Using the condition $ 6s \le n $ and Theorem \ref{th:conLay}, $ L_1 $ is the only layer that connected only with layers $ s-1 $ and $ s+1 $. So, $ L_1 $ is an invariant.

			\end{proof}
			
	\section{The oracle problem of antipodal vertex search} \label{sec:problem}
		\subsection{ Definition }
			
			\begin{definition}
				Let  $ G=(V,E) $ be a regular graph.
				Let $ |V| \le 2^n $.
			    Let $ f: V \to \abkf{0,1}^n $ be the mapping of vertexes to names. The number of possible vertexes names are exponentially larger than number of vertexes.
				
				Let $ h $ be a numeration of neighbours for each vertex. (Holds that if $ u $ is $ k $--th neighbour of $ v $ then $ v $ is $ k $--th neighbour of $ u $).
				
				The algorithm can send to the oracle name of a vertex and a number. The oracle gives the name of a neighbour vertex corresponding to the number. If the name or the number are invalid, then oracle returns an empty string.
				
				The machine solves antipodal vertex search problem if for a given vertex name and oracle it will find the name of any antipodal vertex.
			\end{definition}

			According to the previous section, if $ s < n/6 $ then each vertex of \CCay{s} has exactly one antipodal one.
			
			In the sequel we call the oracle problem of antipodal vertex search as the problem.
			
			We call the number of queries to oracle as complexity of an algorithm.
			
			\begin{remark}
				The mapping $ f $ inducts the isomorphism of graphs $ G=(V,E) $ and $ f(G)=(f(V),f(E)) $
			\end{remark}
			 
		\subsection{ Classical algorithm} \label{sec:clAlg}
			Let $ s < n/6 $.
			
			We present an algorithm that solves the problem with probability $ \frac{1}{m} $ and complexity  $O(m^2 \ \frac{n}{s})$.
			
			Let $ v_0 $ be the initial vertex. Without loss of generality, we can assume that $ v_0 $ is the image of the vertex $ 0^n $.
			
			\begin{itemize}
			\item Initialisation:
			\begin{enumerate}
				\item
				The algorithm makes $ m $ queries to get all neighbours of $ v_0 $ . The oracle answers are exactly $ f(L_s) $. 
				\item
				The algorithm takes $ v_1 \in f(L_s)$.
				\item
				For each $ u$ such that $(u,v_1) \in f(E) $ the algorithm  finds out the weight of $ f^{-1}(u) $ with following actions:
					\begin{enumerate}
						\item
						Get all neighbours of $ u $.
						\item
						Calculate $ x $ as the number of neighbours from $ f(L_s) $.
						\item
						According to Lemma \ref{lm:connum}: $ k^{-1}(x)=|f^{-1}(u)| $.
					\end{enumerate}
				The operation requires $ m^2 $ queries.
				\item
				$ t := 1 $.
			\end{enumerate}
			\item
			At step $ t $ the algorithm keeps in memory:
			\begin{enumerate}
				\item
				The name of vertex $ v_t :\ |f^{-1}(v_k)| = t s \ $.
				\item
				The set $ N(v_t) = \abkf{u \in f(V)|\ (u,v_t) \in f(V)} $.
				\item
				$ |f^{-1}(u)| $ for each $ u \in N(v_t) $.
			\end{enumerate}
			
			\item
			If $ (t+1) s \le n $:
			\begin{enumerate}
				\item
				The algorithm takes $ v_{t+1} \in N(v_t) $ such that $ |f^{-1}(v_{t+1})| = (t+1) s $.
				\item
				The algorithm makes $ m $ queries to build $ N(v_{t+1}) $.
				\item
				For each $ u \in N(v_{t+1}) $ the algorithm calculates
					\begin{equation*}
					j(u) = \min_{x \in N(v_t) \cap N(v_{t+1})} |f^{-1}(x)|.
					\end{equation*}
				This calculation requires $ m $ queries (for each $ u $).
				
				According to Theorem \ref{th:locCon} and Theorem \ref{th:conLay}: 
				\begin{equation*}
					j(u) = |f^{-1}(u)|-s.
				\end{equation*}
				\item
				$ t := t+1 $.
			\end{enumerate}
			
			\item
			If $ (t+1) s = n $:
			
			The algorithm returns $ v_t $.
			
			\item
			If $ (t+1) s > n $:
			\begin{enumerate}
				\item
				The algorithm takes $ v_{t+1} \in N(v_t) $ such that $ |f^{-1}(v_{t+1})| = n-s $.
				\item
				The algorithm reads a random neighbour of $ v_{t+1} $ and returns it.
			\end{enumerate}
			
			\end{itemize}
						
			It is easy to see that 
			\begin{itemize}
			\item
				The number of steps $ \le \frac{n}{s} $ and each step requires $ m^2 + m $ queries.
			\item
				If $ s \del n $ then the algorithm solves the problem.
			\item
				If $ s \ndel n $ then the algorithm solves the problem with probability $ \frac{1}{m} $.
			\end{itemize}

		\subsection{ Quantum algorithm } \label{sec:quantAlg}
			In quantum case we the have space $\C^{m} \otimes \C^{2^n}$ and the oracle is the operator $ \ShiftOp $ such that:
			\begin{itemize}
				\item
				$ \ShiftOp \ket{b,v} = 0 $ if $ v $ is not a valid name of vertex.
				\item
				$ \ShiftOp \ket{b,v} = \ket{b,f_b(v)} $ otherwise.
				\item
				$ \ShiftOp $ is linear.
			\end{itemize}
			Let $ L $ be the $ |V| $-dimensional subspace of $\C^{2^n}$ generated by valid names.
			Then $ \ShiftOp $ is a shift operator on $\C^{m} \otimes L$.
			
			The quantum algorithm:
			\begin{enumerate}
				\item
				Prepare symmetric initial state $ \ket{\psi_0} = \ket{\Psi}\ket{v_0} $, where $ v_0 $ is the input.
				\item
				Calculate $ \ket{\psi_t} = \ShiftOp \circ ( \hat{\Gamma} \otimes \hat{I}) \psi_{t-1} $ for $ t = 1\ldots T $.
				\item
				Measure the state $ \ket{r}=\ket{a}\ket{v} $. Return $ v $.
			\end{enumerate}
			
			If conditions of Theorem \ref{th:hitTime} hold, then the algorithm requires $ \approx \frac{\pi}{2}m $ queries and the probability to success $ \ge \frac{1}{s^2} - o(1)  $.

	\appendix	
	\section{Kravchuk coefficients estimation: the calculations} \label{sec:kravApp}
		The proof begins in the section \ref{sec:Krevchuk}.
		\begin{proof}
			Recall that (see Equation(\ref{eq:reallKrav}))
			\begin{eqnarray*}
				& \abkm{P(A_n^s)-P(B_n^s)} =
					\abkm{\suml_{l=0}^s P(A_n^s|C_l) P(C_l)- \suml_{l=0}^s P(B_n^s|C_l) P(C_l)} = & \\
				& \abkm{\suml_{l=0}^s \abks{P(A_n^s|C_l) - P(B_n^s|C_l)} P(C_l)} \le & \\
					& \suml_{l=0}^s \abkm{P(A_n^s|C_l)-P(B_n^s|C_l)}P(C_l) = & \\
					\nonumber \\ 
				& \suml_{l=0}^s \abkm{ P(A_{2k}^{l})-P(B_{2k}^l) }P(C_l)=&
			\end{eqnarray*}
				
			\begin{eqnarray*}
				& \suml_{2 | l,\ l=0}^s \frac{\binom{k}{l/2}}{\binom{2k}{l}}
						\frac{\binom{2k}{l} \cdot \binom{n-2k}{s-l} } { \binom{n}{s}}  \chi[ s-l \le n-2k] = &\\
						\nonumber \\
				& \suml_{2 | l,\ l=0}^s 
					\frac{\binom{k}{l/2} } { \binom{n}{s}} \binom{n-2k}{s-l} \chi[ s-l \le n-2k] \le &
			\end{eqnarray*}
			\begin{eqnarray*}
				& \suml_{2 | l,\ l=0}^s 
					\frac{\binom{k}{l/2} } {\binom{n}{s}} \binom{n\delta}{s-l} \le 
				\suml_{2 | l,\ l=0}^s 
					\frac{k^{l/2}}{\binom{n}{s}} (n\delta)^{s-l} \le &\\ & \nonumber& \\
				& \suml_{2 | l,\ l=0}^s 
					\frac{s! \ k^{l/2}}{(n-s)^s} (n\delta)^{s-l} \le 
				\suml_{2 | l,\ l=0}^s 
					\frac{s! \ (n/2)^{l/2}}{n^s} (n\delta)^{s-l} = & \\ & \nonumber& \\
				& \suml_{2 | l,\ l=0}^s 
					\abk{1-\frac{s}{n}}^{-s} \frac{s!}{2^{l/2}} \frac{\delta ^ {s-l} }{n^{s-l/2}} \le & \\ & \nonumber& \\ &
				\abk{1-\frac{s}{n}}^{-s} 
					\suml_{2 | l,\ l=0}^s  \frac{s!}{2^{\frac{l} {2}}}
					\frac{ 2^\frac{s-l}{2} f(n)^{\frac{s-l}{2}} n}{ n^{\frac{s-l}{2}}} \frac{n^{s-l}}{n^{s-l/2}} = & \\ & \nonumber& 
			\end{eqnarray*}
			\begin{eqnarray*}
				& \abk{1-\frac{s}{n}}^{-s} \suml_{2 | l,\ l=0}^s
\frac{2^\frac{s}{2}\ s!\ f(n)^{\frac{s-l}{2}} n} {2^{l}\ n^{s/2} } \le 
				\abk{1-\frac{s}{n}}^{-s} \frac{2^\frac{s}{2}\ (s+1) \cdot s!\ f(n)^{\frac{s}{2}} n} {2 \ n^{s/2} } \le &
			\end{eqnarray*}
			
			\begin{eqnarray*}
				& \le \abk{1-\frac{s}{n}}^{-s} 2 (s+1)! \abk{\frac{ f(n) }{n} }^{s/2} = O\abk{(s+1)! \delta^{s}}. &
			\end{eqnarray*}
			
		\end{proof}

	\section{Concurrent measurement quantum walk } \label{sec:measuredApp}
		In this section we establish some properties of a measured quantum walk in \CCay{s}. The main result of this section is Theorem \ref{th:retAbs}. It states that the probability to return in measured quantum walk is at least $ \frac{1}{poly(m)} $.
			
			\begin{definition}
				$\ket{x}$--measured from time ~$T_0$ walk:
				
				If $t \le T_0$ the state of the system is $\ket{\psi_t} = \WalkOp^t \ket{\psi_0}$.
				
				If $t > T_0$ the state of the system is:
					\begin{eqnarray*}
						 \ket{\psi_{t}} = \WalkOp (I - \Pi_x)  \ket{\psi_{t-1}}= \WalkOp (I - I \otimes \ket{x}\bra{x}) \ket{\psi_{t-1}}
					\end{eqnarray*}
			\end{definition}
			
			\begin{definition}
				The probability to stop at moment $t > T_0$ is $q_t = \abkm{\Pi_x \ket{\psi_{t-1}}}^2 $. Let's define $q_t = 0$ for $t \le T_0$ .
			\end{definition}
			
			\begin{definition}
				The probability to stop till time $t$ is $p_t = \suml_{t'=0}^t q_{t'}$.
			\end{definition}
			
			In this section we prove the bound on the probability to stop in $ \ket{0^n} $--measured walk (and call it as the probability to return).
			
			Let's denote:
			
			$\ket{\xi_t} = \WalkOp^t \ket{\psi_0}$,

			$\alpha_t = \bkScal{\psi_0}{\xi_t} = \bra{\psi_0} \WalkOp^t  \ket{\psi_0}$.
			
			Using result of Section \ref{sec:symCaseProb} we get $\alpha_t \ket{\psi_0} = \Pi_0 \ket{\xi_t}$.
			
			\begin{lemma}
				The following holds:
				
				\begin{eqnarray*}
					&\ket{\psi_{T_0 + \Delta t}} = \ket{\xi_{T_0 + \Delta t}} - \suml_{k=0}^{\Delta t - 1} \beta_{k} \hat{Q}^{\Delta t - k} \ket{\psi_{0}}, & \\
					& \beta_k = \alpha_{T_0 + k} -  \suml_{j = 1}^k \beta_{k-j} \alpha_j .&
				\end{eqnarray*}
			\end{lemma}
			\begin{proof}
				The proof by Induction. The base $\Delta t = 0$ is clearly by the definition: if $t \le T_0$ then  $\ket{\xi_t} = \ket{\psi_t}$.
				
				Acting $\WalkOp (I - \Pi_0)$ on $ \ket{\xi_{T_0 + \Delta t}} - \suml_{k=0}^{\Delta t - 1} \beta_{k} \hat{Q}^{\Delta t - k} \ket{\psi_{0}}$ we get:
				\begin{eqnarray*}
					&\WalkOp (I - \Pi_0) \abks{\ket{\xi_{T_0 + \Delta t}} - \suml_{k=0}^{\Delta t - 1} \beta_{k} \WalkOp^{\Delta t - k} \ket{\psi_{0}}} = &\\
					&\WalkOp \abks{\ket{\xi_{T_0 + \Delta t}} - \alpha_{t+\Delta t} \ket{\psi_{0}} - 
					\suml_{k=0}^{\Delta t - 1} \beta_{k} \WalkOp^{\Delta t - k} \ket{\psi_{0}}
					+
					\suml_{k=0}^{\Delta t - 1} \beta_{k} \alpha_{\Delta t - k} \ket{\psi_{0}}
					} = &
					~\\
					&\ket{\xi_{T_0 + \Delta t + 1}} - 
					\suml_{k=0}^{\Delta t - 1} \beta_{k} \WalkOp^{\Delta t - k + 1} \ket{\psi_{0}}
					+
					\abk{ -\alpha_{t+\Delta t}  + \suml_{k=0}^{\Delta t - 1} \beta_{k} \alpha_{\Delta t - k} } \WalkOp \ket{\psi_{0}}
					 & ~ \\
					 &=\ket{\xi_{T_0 + \Delta t + 1}} - \suml_{k=0}^{\Delta t - 1} \beta_{k} \WalkOp^{\Delta t - k + 1} \ket{\psi_{0}} - \beta_{\Delta t + 1} \WalkOp^1 \ket{\psi_0} = &~\\
					 &\ket{\xi_{T_0 + (\Delta t + 1)}} - \suml_{k=0}^{(\Delta t + 1) - 1} \beta_{k} \WalkOp^{(\Delta t + 1) - k} \ket{\psi_{0}}.&
				\end{eqnarray*}
			\end{proof}
			
			\begin{coroll} { ~ }
			
				$\beta_{\Delta t} = \Pi_0 \ket{\psi_{T_0 + \Delta t}} = \bkScal{\psi_0}{\psi_{T_0 + \Delta t}}$,
				
				$q_{T_0 + \Delta t} = |\beta_{\Delta t}|^2$.
			\end{coroll}
			
			We suppose here that \textbf{$s$ is odd}. This condition is necessary and sufficient to have $\Pi_0 \ket{\psi_{2t+1}}=0, \ \forall t$. Also we suppose that $T_0$ is even.
			
			\begin{lemma} \label{lem:diffAlphaEst} ~
				
				Let
				$t = O(m)$,
				
				$2d_k^s / m = O(1/n)$ and $\abkm{k-n/2}/n \le \delta= \sqrt{\frac{2\ln n}{n}}$.
				
		 		Then $ \abkm{\alpha_{2t} - \alpha_{2(t+1)} } \le O(\frac{1}{n}) $.
			\end{lemma}
			
			\begin{proof}
				Using (\ref{eq:retEst}) we get:
				$$ \alpha_{t} = \suml_{k=0}^n{
						\frac{1}{2^n} \binom{n}{k}
							\cos{ (\omega_k^s  t)}
					}.  \eqno{( \ref{eq:retEst})} $$
				
				Let $b_k = \frac{\pi}{2} - \omega_k^s$, using Chernoff estimation we get:	
				\begin{eqnarray*}
					&\frac{1}{2^n} \abkm{ \suml_{k=0}^n \abks{ \binom{n}{k} \cos{ (\omega_{d_k}  2t)} -
					\binom{n}{k} \cos{ (\omega_{d_k}  2(t+1)} }} \le &\\
					&\frac{1}{2^n} \abkm{ \suml_{k \in M} \binom{n}{k}  (-2) \abks{ \sin \abk{ \omega_{d_k}  (2t+\frac{1}{2}) } \sin \abk{\omega_{d_k}}}} + O(\frac{1}{n}) \le& \\
					&\frac{2}{2^n} \abkm{ \suml_{k \in M} \binom{n}{k}  \sin \abk{\omega_{d_k}}} + O(\frac{1}{n}) =& \\
					&\frac{2}{2^n} \abkm{ \suml_{k \in M} \binom{n}{k}  \cos \abk{b_k + O(b_k^3)}} + O(\frac{1}{n}) \le O(\frac{1}{n}).&
				\end{eqnarray*}
			\end{proof}
			
			Using previous results we get a bound on $\beta_{2t}$:
			\begin{eqnarray*}
				&\beta_{2t} = \alpha_{T_0 + 2t} -  \suml_{j = 1}^t \beta_{2(t-j)} \alpha_{2j} =&\\
		  		&\alpha_{T_0 + 2t} - \suml_{j = 0}^{t-1} \beta_{2(t-j-1)} \alpha_{2j+2} + \suml_{j = 0}^{t-1} \beta_{2(t-j-1)} \alpha_{2j} - \suml_{j = 0}^{t-1} \beta_{2(t-j-1)} \alpha_{2j} = &\\
		 		&\alpha_{T_0 + 2t} - \suml_{j = 0}^{t-1} \beta_{2(t-j-1)} \alpha_{2j} - \suml_{j = 0}^{t-1} \beta_{2(t-j-1)} \abk{ \alpha_{2j} - \alpha_{2j+2}} &
			\end{eqnarray*}
			
			Rewriting the first equation for $\beta_{T_0+2t-2}$ and using $\alpha_0 = 1$, we obtain $\alpha_{T_0 + 2t - 2} = \suml_{j = 0}^{t-1} \beta_{2(t-j-1)} \alpha_{2j}$.
			
			\begin{lemma}
				In conditions of lemma \ref{lem:diffAlphaEst}. If $s$ is odd, $T_0$ is even then	
				\begin{eqnarray*}
					\abkm{\beta_{2t}} \ge |\alpha_{T_0 + 2t}| - |\alpha_{T_0 + 2t-2}| - O(\frac{1}{n}) \suml_{k=0}^{t} |\beta_{2k}|
				\end{eqnarray*}			
			\end{lemma}
			
			We want to find the return probability till moment $T$, starting with $T_0 \le T$. Let $0 \le \tau \le T-T_0$ is such, that $|\alpha_{T_0 + \tau}| \le \gamma < 1$.
			
			\begin{lemma}
				If $\tau$ is such that $\suml_{k=0}^{T - T_0 - \tau} |\beta_k|  = o(f) $ then
				
				either $p_T = \Omega( f / \sqrt{\tau} ) $
				
				or $\forall t: T_0 \le 2t \le T \quad |\beta_{2t}| \le |\alpha_{T_0 + 2t}| - |\alpha_{T_0 + 2t-2}| - o(f/n) $
			\end{lemma}
			\begin{proof}
				\begin{eqnarray*}
					\suml_{k=0}^{T - T_0} |\beta_k| = \suml_{k=0}^{T - T_0 - \tau} |\beta_k| + \suml_{k=0}^{\tau-2} |\beta_k|.
				\end{eqnarray*}	
				If wrong inequality $\suml_{k=0}^{T - T_0} |\beta_k| = o(f)$
				then $\suml_{k=0}^{\tau-2} |\beta_k| = \Omega(f)$
				Using Cauchy–Schwarz inequality we get:
				\begin{eqnarray*}
					\suml_{k=0}^{\tau-2} |\beta_k| \le \sqrt{\tau} \sqrt{ \suml_{k=0}^{T - T_0} |\beta_k|^2} = \sqrt{\tau} \sqrt{ p_T}.
				\end{eqnarray*}	
				and the second variant holds.
				
				If holds $\suml_{k=0}^{T - T_0} |\beta_k| = o(f)$, the previous lemma proves the first variant.
			\end{proof}
			
			\begin{lemma} {~}
				\label{th:poglest2}
				
				Let $|\alpha_T| - q > 0$.
				
				If $\tau$ is such that $\suml_{k=0}^{T - T_0 - \tau} |\beta_k|  = o(f)$ and $(\tau f /n )=o(1)$ then
				
				$ p_T \ge \frac{(|\alpha_T| - q - o(1))^2}{\tau}$.
			\end{lemma}
			\begin{proof}
				If the second variant holds we use the previous lemma:
				\begin{eqnarray*}
					&\suml_{k=0}^{T} |\beta_k|^2 \ge \frac{1}{\tau} \abk{ \suml_{k=0}^{T - T_0 - \tau} |\beta_k|}^2 \ge& \\
					&\frac{1}{\tau} \abk{ \suml_{k=0}^{T - T_0 - \tau} |\alpha_k| - |\alpha_{k-2}| - o(f/n)}^2 = \frac{(|\alpha_T| - q - o(1))^2}{\tau}.&
				\end{eqnarray*}
				
				The second move can be done because the expression under the brackets is positive because of conditions.
				
				The found estimation is worse than the first variant. Therefore the lemma is proved.
								
			\end{proof}
			
			Let parameters satisfy the conditions of Theorem \ref{th:retTime} for returning at moment T. Then $|\alpha_T| = 1 - o(1)$.
			
			\begin{lemma}
			
				If $\suml_{k=0}^{T - T_0 - \tau} |\beta_k|  = \Omega(f)$ then $p_T = \Omega( f^2 / (T - T_0 - \tau) )$.
			\end{lemma}
			The proof is omitted (use Cauchy–Schwarz inequality).
			
			Combining above claims we get:
			
			Let $T-\tau$ be a moment the measurement start and the end of the observation. Then holds one of three variants:			
			\begin{itemize}
				\item
					at the starting moment the amplitude is high enough and the stop probability is $\Omega( f^2 / (T - T_0 - \tau) )$
				\item
					the probability to stop till moment $\tau$ is high
				\item
					absorbed a few, and the state at moment $ T $  is almost the same as in nonmeasured walk (so, the probability to stop near moment $ T $ is high).
			\end{itemize}
			
			More formally, we get the theorem:

			\begin{theorem}[Returning with absorption] \label{th:retAbs}~
			
				Let $ s $ is odd.
				
				Let parameters satisfy conditions of Theorem \ref{th:retTime} with $t = T $.
				
				Let parameters satisfy conditions of Theorem \ref{th:hitTime} with $t = T_p $ where $ T_p = T/2 + \epsilon$.
				
				Then the probability to return in quantum $ \ket{0^n} $--measured walk with measurement starting at moment $T_0 \le T_p$ is
				
				$p_T \ge \Omega( \frac{n } { \epsilon (T - T_p)^2 } )$.
			\end{theorem}
			\begin{proof}
				Combining previous statements with following parameters :
				
				\begin{enumerate}
					\item
					$\tau = T/2 - \epsilon = T - T_p$,
					\item
					$|\alpha_{T_p}|^2 \le 1 - p^{hit}_{T_p}$ --- theorem about hitting time,
					\item
					$f = \sqrt{n} / (T - T_p) $ --- follows from conditions of lemma \ref{th:poglest2}.
				\end{enumerate}
				we get:
				$p_T \ge min( \frac{(1 - \sqrt{1 - p^{hit}_{T_p}} - o(1))^2}{  T - T_p }  ; \Omega( \frac{n } { \epsilon (T - T_p)^2 } ) ) = \Omega( \frac{n } { \epsilon (T - T_p)^2 } ) $.
			\end{proof}
			
			\begin{remark}
				The following satisfy the previous theorem: $\epsilon \le 1; \ T-T_p = \frac{\pi}{2} m$.
			\end{remark}

	\section{Local connections theorem }
		
			\begin{theorem}[Local connections between layers] \label{th:locCon}
			~
			Let $ v,u,q $ are vertices such that $ (v,u), (u,q) \in E $, $ |v| = l,\ |q|=t  $.
			
			Then:
				
			$ x $ is a layer such that $ (l,x), (t,x) \in E'$
			$\Leftrightarrow$
			there exists vertex $w \in L_x$ such that $ (v,w), (q,w) \in E $.
			\end{theorem}	
			\begin{proof}
				To prove the theorem we find each layer $ x $ such that exist vertices $e_1,\ e_2$ of weight $s$ and hold the following:
				\begin{itemize}
					\item
					$|v+e_1| = x$,
					\item
					$ v + e_1 + e_2 = q$.
				\end{itemize}

				Note that Hamming difference between $v$ and $q$ is at most $2s$ and it's even.
				
				Let divide $n$ coordinates on 4 parts:
				\begin{itemize}
					\item the first part of size $a = | v \cap q | $  --- coordinates in which both $ v $ and $ q $ have one.
					\item the second part of size $b_1$ --- coordinates in which $ v $ has one, excluding chosen $ a $ coordinates. ($a+b_1 = l$)
					\item the third part of size $b_2$ --- the same for $ q $. ($a+b_2 = |q|$).
					\item the last part of size $d$ --- the others (both vectors has zeros).
				\end{itemize}
				
				Let $e_1$ has $\alpha$ ones in the first part, $\beta,\ \gamma,\ \delta$ are the number of ones in other parts respectively.
				
				Then $e_2$ must have $\alpha,\ b_1-\beta,\ b_2-\gamma, \ \delta $ ones in corresponding parts: in the first and the fourth --- on the same places there $e_1$ has, in the second and third --- on the other places, there $e_1$ has zeros. (So, $e_2$ is fully defined by $e_1$.)
				
				All parameters satisfy the following inequalities and equations:
				\begin{eqnarray*}
					b_1 + b_2 \le 2s; \\
					\alpha \le a;\ \beta \le b_1; \ \gamma \le b_2; \ \delta \le d; \\
					a + b_1 + b_2 + d = n; \\
					|v| + (\gamma + \delta) - (\alpha+\beta) = |v + e_1|; \\
					\alpha + \beta + \gamma + \delta = |e_1| = s; \\
					\alpha + (b_1-\beta) + (b_2-\gamma) + \delta = =|e_2| = s.
				\end{eqnarray*}

				The last two equations give
				\begin{eqnarray*}
					\beta + \gamma = \frac{b_1+b_2}{2}, \\
					\alpha + \delta = s - \frac{b_1+b_2}{2}.
				\end{eqnarray*} 
				
				It easy to see that if $ (\alpha, \beta, \gamma, \delta) $ satisfy:
				\begin{eqnarray*}
					\alpha \le a;\ \beta \le b_1; \ \gamma \le b_2; \ \delta \le d; \\
					\alpha + \beta + \gamma + \delta = s;
				\end{eqnarray*}
				then exist $ e_1, e_2 $ with corresponding properties.

				The weight of $|v+e_1|$ is minimal if $|v\cap e_1|$ is maximal. This happens exactly then:
				\begin{eqnarray*}
					\alpha = \min(a, s - \frac{b_1+b_2}{2}), \\
					\beta = \min(b_1, \frac{b_1+b_2}{2}).
				\end{eqnarray*}
				
				So we get the maximum value of $ x $:
				\begin{eqnarray*}
					& \max x = (a + b_1) + (\gamma + \delta) - (\alpha+\beta) = a + b_1 + s - 2(\alpha + \beta) = &\\
					& a + b_1 + s - 2(\min(a, s - \frac{b_1+b_2}{2}) + \min(b_1, \frac{b_1+b_2}{2})) = & \\
					& \max(s - a - b_1, a + b_1 - s, s - a - b_2, a + b_2 -s ) = &\\
					&\max(|l-s|, |t-s|) = A.&
				\end{eqnarray*}
				
				The weight of $|v+e_1|$ is maximal if $|v\cap e_1|$ minimal. This happens exactly then:
				\begin{eqnarray*}
					\delta = \min(d, s - \frac{b_1+b_2}{2}), \\
					\gamma = \min(b_2, \frac{b_1+b_2}{2}).
				\end{eqnarray*}
				
				So we get the minimal value of $ x $:			
				\begin{eqnarray*}
					&|v+e_1| = (a + b_1) + (\gamma + \delta) - (\alpha+\beta) = a + b_1 - s + 2(\gamma + \delta) = &\\
					& a + b_1 - s + 2( \min(b_2, \frac{b_1+b_2}{2}) +\min(d, s - \frac{b_1+b_2}{2})) = & \\
					& \min( (a+b_1+b_2+d)+(b_2+d-s), a + b_1 + s, a+b_2+s,& \\ &(a+b_1+b_2+d)+(b_2+d-s) ) = &\\
					&\min( \min(l+s, 2n-l-s) , \min(t+s, 2n-t-s)) = B.&
				\end{eqnarray*}
				
				Let $ (\alpha, \beta, \gamma, \delta) $ be any valid set of parameters corresponding to vectors $ e_1, e_2 $. Note that:
				\begin{itemize}
					\item
					If the set $ (\alpha', \beta', \gamma', \delta')=(\alpha-1, \beta, \gamma, \delta+1) $ is valid then $ |v+e_1'|=2+|v+e_1| $.
					\item
					If the set $ (\alpha'', \beta'', \gamma'', \delta'')=(\alpha, \beta-1, \gamma+1, \delta) $ is valid then $ |v+e_1''|=2+|v+e_1| $.
					\item
					If $ |v+e_1| \ne A $ then at least one of this sets are valid.
				\end{itemize}

				In other words, $B, B+2, B+4, \ldots, A-2, A $ are possible values of x.
				
				So, we get:
				\begin{eqnarray}
					\nonumber \abkf{x| \exists w \in L_x, \ (v,w),(v,q) \in E } =  \\ \{x \in [B, A]\ |\  x \equiv |l-s| \pmod{2}\} . \nonumber				\end{eqnarray}
				
				According to Theorem \ref{th:conLay}:
				$$
					\abkf{x| (x,l), (x,t) \in E'} = \{x \in [B, A]\ |\  x \equiv |l-s| \pmod{2}\}.
				$$
				
				The set of layers satisfying the left condition is equal to the set of layers satisfying the right condition.
			\end{proof}

\end{document}